\documentclass[12pt]{article}
\usepackage{amsmath,amsfonts,amsthm,amssymb}
\usepackage{url}
\usepackage{graphicx}
\usepackage[font=small]{caption}
\usepackage{subcaption}
\usepackage[colorlinks=true]{hyperref}
\hypersetup{
    linkcolor = {blue},
    citecolor = {blue}
}
\usepackage{cleveref}
\usepackage{enumitem}
\usepackage{algpseudocode}
\usepackage{algorithmicx}
\usepackage{bbm}
\usepackage{array}
\usepackage{xfrac}
\usepackage{tikz}
\usepackage{float}
\usepackage{enumitem}
\usepackage{color}
\usepackage{mathtools}
\usepackage{comment}
\usepackage[a4paper, total={6.5in, 9in}]{geometry}
\usepackage{thm-restate}
\usepackage{makecell}

\graphicspath{{images/}}

\newcommand{\nm}[0]{N}

\newcommand{\be}[0]{\begin{enumerate}}
\newcommand{\ee}[0]{\end{enumerate}}
\newcommand{\bi}[0]{\begin{itemize}}
\newcommand{\ei}[0]{\end{itemize}}

\newcommand{\OO}{\mathrm{O}}

\newcommand{\Decreasekey}{\mbox{\it decrease-key}}
\newcommand{\Delete}{\mbox{\it delete}}

\newcommand{\Deletemin}{\mbox{\it delete-min}}

\newcommand{\Findmin}{\mbox{\it find-min}}

\newcommand{\Insert}{\mbox{\it insert}}

\newcommand{\Makeheap}{\mbox{\it make-heap}}
\newcommand{\Meld}{\mbox{\it meld}}

\newtheorem{theorem}{Theorem}
\newtheorem{lemma}[theorem]{Lemma}
\newtheorem{corollary}[theorem]{Corollary}

\newtheorem{remark}[theorem]{Remark}

\title{A Simpler Proof that Pairing Heaps Take $\OO(1)$ Amortized Time per Insertion}
\author{Corwin Sinnamon\thanks{Department of Computer Science, Princeton University.  Research partially supported by a gift from Microsoft.} \and Robert Tarjan\footnotemark[1]}
\date{}
\begin{document}
\maketitle

\begin{abstract}
The \emph{pairing heap} is a simple ``self-adjusting'' implementation of a heap (priority queue).  Inserting an item into a pairing heap or decreasing the key of an item takes $\OO(1)$ time worst-case, as does melding two heaps.  But deleting an item of minimum key can take time linear in the heap size in the worst case.  The paper that introduced the pairing heap~\cite{FSST86} proved an $\OO(\log n)$ amortized time bound for each heap operation, where $n$ is the number of items in the heap or heaps involved in the operation, by charging all but $\OO(\log n)$ of the time for each deletion to non-deletion operations, $\OO(\log n)$ to each.  Later Iacono~\cite{IaconoPairing} found a way to reduce the amortized time per insertion to $\OO(1)$ and that of meld to zero while preserving the $\OO(\log n)$ amortized time bound for the other update operations.  We give a simpler proof of Iacono's result with significantly smaller constant factors.  Our analysis uses the natural representation of pairing heaps instead of the conversion to a binary tree used in the original analysis and in Iacono's.
\end{abstract}

\section{Introduction}
\label{S:introduction}

A \emph{heap} (or \emph{priority queue}) is a data structure containing a set of items, each with an associated key selected from a totally ordered key space.  Heaps support the following operations:

\begin{itemize}
\item[] $\Makeheap()$: Create and return a new, empty heap.

\item[] $\Findmin(H)$: Return an item of smallest key in heap $H$; return null if $H$ is empty.

\item[] $\Insert(H, x)$: Insert item $x$ with predefined key into heap $H$.  Item $x$ must be in no other heap.

\item[] $\Deletemin(H)$: Delete from $H$ the item that would be returned by $\Findmin(H)$, and return it.

\item[] $\Meld(H_1, H_2)$: Return a heap containing all items in item-disjoint heaps $H_1$ and $H_2$, destroying $H_1$ and $H_2$ in the process.

\item[] $\Decreasekey(H, x, k)$: Given the location of item $x$ in heap $H$, and given that the key of $x$ is at least $k$, decrease its key to $k$. 

\item[] $\Delete(H, x)$: Given the location of item $x$ in heaps $H$, delete $x$ from $H$. 
\end{itemize}

One can implement $\Delete(H, x)$ as $\Decreasekey(H, x, -\infty)$ followed by $\Deletemin(H, x)$.  We shall assume this implementation and not further mention $\Delete$: To within a constant factor its time bound is the same as that of $\Deletemin$.  Henceforth by \emph{deletion} we mean a $\Deletemin$ operation.

If the only operations on keys are comparisons, $n$ insertions followed by $n$ deletions will sort $n$ items by key, so the amortized time\footnote{We shall study \emph{amortized time} throughout. See~\cite{tarjan1985amortized} for information on amortized analysis.} of either insertion or deletion on an $n$-item heap must be $\Omega(\log n)$.  The \emph{Fibonacci heap}~\cite{Fibonacci}  supports deletion in $\OO(\log n)$ amortized time and all other heap operations in $\OO(1)$ amortized time.  Since the invention of the Fibonacci heap, many other heaps with similar efficiency have been developed.  See~\cite{BrodalSurvey, HollowHeaps, KS19}.

One is the \emph{pairing heap}~\cite{FSST86}, a ``self-adjusting" heap that is simple and efficient in practice~\cite{LarkinSenTarjan, StaskoVitter}.  The paper that introduced pairing heaps proved an $\OO(\log n)$ amortized time bound for all operations.  The authors conjectured that pairing heaps have the same amortized efficiency as Fibonacci heaps, but this was disproved by Fredman~\cite{FredmanLB}, who showed that pairing heaps and similar self-adjusting heaps must take $\Omega(\log\log n)$ time per decrease-key if they take $\OO(\log n)$ time per deletion.  Iacono and {\"O}zkan~\cite{IaconoOzkan} proved the same lower bound for a different large class of self-adjusting heaps. These results raise the question of whether pairing heaps, or any other kind of self-adjusting heap, have efficiency matching these lower bounds.  The question remains open for pairing heaps, but we have answered it in the affirmative for two other heap implementations, slim and smooth heaps~\cite{phd2022}, as we discuss in Section~\ref{S:remarks}.

Improving the bound for decrease-key is not our goal here.  Rather, it is to reduce the bound for insertions and melds.  Iacono~\cite{IaconoPairing} reduced the amortized time bound per insertion in pairing heaps to $\OO(1)$ and that of meld to zero while preserving the $\OO(\log n)$ bound for decrease-keys and deletions.  His proof is complicated and has large constant factors, however.  We give a simplified proof of Iacono's result with much smaller constant factors.  That is, we prove that pairing heaps have the following amortized time bounds: $\OO(1)$ per insertion, zero per meld, and $\OO(\log n)$ per decrease-key and deletion.  Our approach applies to other self-adjusting heap implementations, as we discuss in Section~\ref{S:remarks}.

\section{Pairing heaps}
\label{S:pairing-heaps}

A pairing heap stores the items in a heap as the nodes of an ordered tree.  Henceforth we shall speak of nodes rather than items.  The tree is heap-ordered by key: If node $x$ with key $x.key$ is the parent of node $y$ with key $y.key$, then $x.key \leq y.key$.  Thus the root is a node of minimum key.  Access to the tree is via the root.

The fundamental primitives for modifying heap-ordered trees are \emph{linking} and its inverse, \emph{cutting}.  A \emph{link} of the roots of two node-disjoint heap-ordered trees combines the trees by making the root of smaller key the parent of the other root, breaking a tie arbitrarily.  The surviving root is the \emph{winner} of the link; the new child is the \emph{loser} of the link.  We denote by $xy$ a link won by $x$ and lost by $y$.

A link makes the loser the new \emph{first} child of the winner. Thus the children of any node are ordered by link time, latest first.  We think of this order as being left to right: the first and last child are \emph{leftmost} and \emph{rightmost}, respectively; the \emph{left} and \emph{right} siblings of a child $y$ of $x$ are the children linked to $x$ before and after $y$ was linked to $x$, respectively.

A \emph{cut} of a link $xy$ breaks the link, breaking the tree containing $x$ and $y$ into two trees, one rooted at $y$ containing the entire subtree of $y$ as it existed before the cut, and one containing $x$ that is unchanged except for the removal of $y$ and its descendants.

A pairing heap does the heap operations as follows:

\begin{itemize}
\item[] $\Makeheap()$: Create and return an empty tree.

\item[] $\Findmin(H)$: Return the root of $H$.

\item[] $\Insert(H, x)$: Make $x$ into a one-node tree.  If $H$ is non-empty, link $x$ with the root of $H$; otherwise, make $x$ the root of $H$.

\item[] $\Meld(H_1, H_2)$: If either $H_1$ or $H_2$ is empty, return the other; otherwise, return the tree formed by linking the roots of $H_1$ and $H_2$.

\item[] $\Decreasekey(H, x, k)$: Set $x.key = k$.  If $x$ is not the root of $H$, cut the link lost by $x$ and link $x$ with the root of $H$.

\item[] $\Deletemin(H)$: Let $x$ be the root of $H$.  If $H$ has no other nodes, replace $H$ by an empty tree and return $x$.  Otherwise, cut all the links between $x$ and its children, making the list of children of $x$ (ordered by link time) into a list of roots.  Link these roots in two passes.  The first, the \emph{pairing pass} links the roots in pairs left to right, the first with the second, the third with the fourth, and so on.  The second, the \emph{assembly} pass, repeatedly links the current rightmost root with its left neighbor until only one root remains.  Replace $H$ by the tree rooted at the remaining root, and return $x$. See Figure~\ref{F:pairing-heap}.
\end{itemize}

\begin{figure}[ht!]
\centering
\includegraphics[width=3.8in]{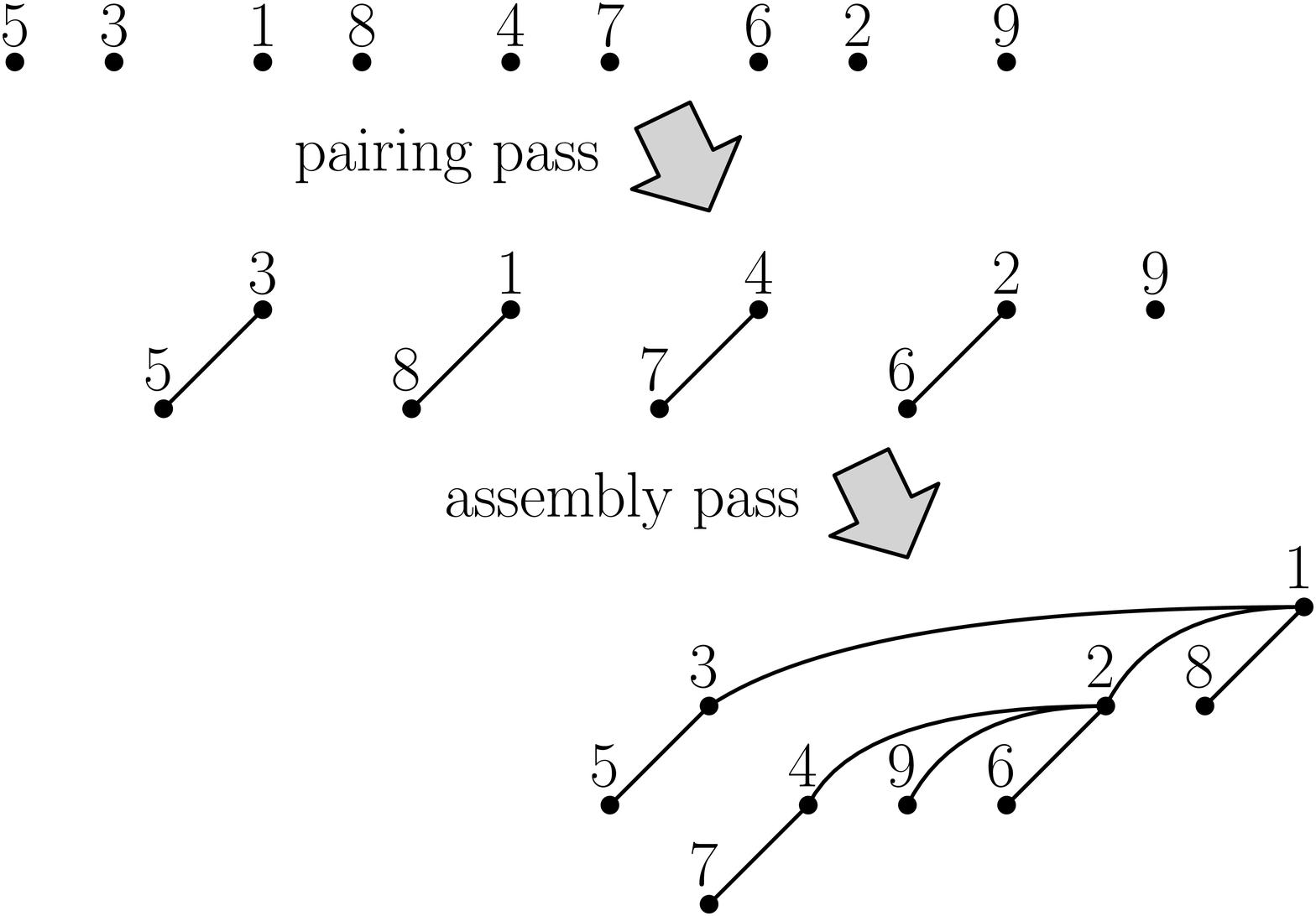}
\caption{Pairing heap linking during delete-min, after the root is deleted.  The circles represent nodes and the numbers indicate their keys.  Lines between nodes represent links, with the winner above the loser.}
\label{F:pairing-heap}
\end{figure}

One good way to represent a pairing heap is to use three pointers per node: one to the leftmost child of a node, one to its right neighbor, and one to its left neighbor, or to its parent if it has no left neighbor.  The first pointer of a node is null if it has no children; the second is null if the node is a rightmost child or a root; the third is null if the node is a root.  This representation allows a link or cut to be done in $\OO(1)$ time and makes the worst-case time per operation $\OO(1)$ plus the number of links done.

A more compact representation with only two pointers and one bit per node~\cite{FSST86} uses the first pointer to indicate the leftmost child of a node, or its right neighbor if it has no children.  If a node is a leftmost child, its second pointer indicates the right neighbor of its parent, or its parent if it has no right siblings; if a node is not a leftmost child, its second pointer indicates its left neighbor.  Pointers are null when the designated nodes do not exist.  This representation is ambiguous in one case, when the first pointer of a node $x$ indicates $y$ and the second pointer of $y$ indicates $x$.  Node $x$ could have no right siblings and leftmost child $y$, or $x$ could have no children and right neighbor $y$.  To disambiguate the representation in this case, each node has a bit indicating whether it is a leftmost child.  Only the root has a null second pointer.

A third representation uses \emph{hollow nodes}~\cite{HollowHeaps}.  We make the tree exogenous rather than endogenous~\cite{tarjanbook}: Instead of the nodes \emph{being} the items, they \emph{hold} the items, or pointers to the items.  Each node has a pointer to its leftmost child and to its right neighbor.  To do a decrease-key of an item not in the root, we allocate a new node and move the item whose key decreases to this new node, making its old node \emph{hollow}.  We move the children of the old node to the new node.  During a deletion, we deallocate any hollow nodes that become roots.  Each item maintains a pointer to the current node holding it.  The hollow-node representation is simple, but it uses at least as much space as either of the endogenous representations, and if there are many decrease-key operations a tree can become mostly hollow (although one can do periodic deallocation of hollow nodes to overcome this). It is simple though, and it may be appropriate in situations that require an exogenous representation.    

We consider a sequence of pairing heap operations on an initially empty collection of heaps.  Each operation except deletion does at most one link and takes $\OO(1)$ time.  Deletion takes $\OO(1)$ time plus $\OO(1)$ time per link. Hence to bound the total time of the sequence we just need to bound the number of links.

\section{Node and link types}
\label{S:node-and-link-types}

We call a node \emph{temporary} if it is eventually deleted, \emph{permanent} if not.  We state our time bounds in terms of the number of temporary nodes.  Given a heap operation, we denote by $n$ the number of temporary nodes in the heap or heaps undergoing the operation.  In stating time bounds we assume $n \geq 4$; if $n < 4$ in a given operation the amortized time of the operation is $\OO(1)$.

To bound the number of links, we classify links in four overlapping ways.  We call a link an \emph{insertion link}, \emph{decrease-key link}, \emph{pairing link}, or \emph{assembly link} if it is done by an insertion or meld, by a decrease-key, by the pairing pass of a deletion, or by the assembly pass of a deletion, respectively. 

\begin{lemma}
\label{L:insertion-links}
The number of insertion links is at most the number of insertions.
\end{lemma}
\begin{proof}

Consider the number of non-empty trees in a collection of heaps. This number is initially zero, always non-negative, decreases by one each time a meld does a link, and can only increase, by one, during an insertion into an empty tree, which does not do a link. Hence the number of melds that do links is at most the number of insertions that do not do links.
\end{proof}

\begin{lemma}
\label{L:assembly-links}
The number of assembly links during a deletion is at most the number of pairing links during the deletion. 
\end{lemma}
\begin{proof}
Consider a deletion that does $k$ links between $k+1$ roots.  It does at least $k/2$ pairing links and hence at most $k/2$ assembly links.
\end{proof}

In addition to classifying links based on the operation that does them, we classify them based on their futures. A link is a \emph{deletion link}, abbreviated \emph{d-link}, if it is cut by a deletion; a \emph{key link}, abbreviated {k-link}, if it is cut by a decrease-key; a \emph{final link}, abbreviated \emph{f-link}, if it is never cut.  

\begin{lemma}
\label{L:f-links}
The number of final links plus the number of deletions is at most the number of insertions.
\end{lemma}
\begin{proof}
Every final link is lost by a different permanent node.  Every deletion deletes a different temporary node.
\end{proof}

We call a link \emph{real} if its winner and loser are both temporary and it is not cut by a decrease-key, \emph{phantom} if its winner or loser is permanent or it is cut by a decrease-key.  We call a child \emph{real} or \emph{phantom} if it is connected to its parent by a real or phantom link, respectively.

We call a link done during a deletion \emph{left} or \emph{right} if the loser is left or right of the winner on the root list before the link, respectively.

By Lemma~\ref{L:assembly-links}, it is enough to bound the number of pairing links.  We do this in three steps.  First, in Section~\ref{S:pairing-links} we show that the number of such links is at most a constant times the number of insertions, decrease-keys, real right assembly links, and real pairing links.  Second, in Section~\ref{S:size} we bound the number of real right assembly links.  Third, in Section~\ref{S:mass} we bound the number of real pairing links.  In Section~\ref{S:total} we combine our bounds to complete our analysis of pairing heaps.

\section{Pairing links}
\label{S:pairing-links}

\begin{theorem}
\label{T:pairing}
The total number of pairing links during a sequence of pairing heap operations is at most four per insertion plus three per decrease-key plus two per real right assembly link plus two per real pairing link.
\end{theorem}
\begin{proof}
Let $vw$ be a pairing link created during the deletion of a node $u$. Node $u$ is temporary, since it is deleted.  Let $uv$ and $uw$ be the links lost by $v$ and $w$ that were cut when $u$ was deleted. Both links were done in earlier operations.  In the following five cases, we charge $vw$ to itself or to a link or operation related to $uv$ or $uw$.  We say $vw$ is a \emph{Case-$i$ link} if Case $i$ is the first that applies to $vw$.

\noindent\textbf{Case 1:}
Link $vw$ is a k-link, f-link, or real link.  We charge $uv$ to itself.  There is at most one Case-1 link per pairing f-link, pairing k-link, and real pairing link.

In the remaining cases $v$ is temporary and $w$ is permanent, since if $v$ were permanent $vw$ would be a k-link or f-link, and if $v$ and $w$ were temporary $vw$ would be a k-link or real link.

\noindent\textbf{Case 2:}
Link $uw$ or $uv$ is an insertion link or a decrease-key link. Charge $vw$ to this link.  Link $vw$ is uniquely determined by either $uv$ or $uw$, since $vw$ is the first link that $v$ and $w$ participate in after $u$ is deleted.

In the remaining cases $uv$ and $uw$ are pairing or assembly links.

\noindent\textbf{Case 3:}
Link $uv$ is a pairing link.  Link $uv$ is real because $u$ and $v$ are temporary and $uv$ was cut by a deletion.  We charge $vw$ to $uv$.  Link $uv$ can only be charged once in this way since $uv$ determines $vw$, so there is at most one Case-3 link per real pairing link.

\noindent\textbf{Case 4:}
Link $uw$ is an assembly link.  Node $w$ must have won a link during the pairing pass that preceded $uw$, or else $w$ was the one root that did not participate in a pairing link.  Since $w$ is permanent, if it won a link, it was a k-link or f-link.  We charge $vw$ to this link.  If $w$ did not win a link, we charge $vw$ to the deletion that did $uw$. There is at most one Case-4 link per pairing k-link, pairing f-link, and deletion.

\noindent\textbf{Case 5:}
Link $uv$ is an assembly link and $uw$ is a pairing link.  Consider the assembly pass that does $uv$.  After winning $uv$, $u$ is the rightmost root, and it either loses a right assembly link during the pass (say $xu$) or it is the sole remaining root at the end of the pass.  Since $u$ is temporary, $xu$ is either a k-link or a real right assembly link.  We charge $vw$ to $xu$ if it exists, to the deletion that does $uv$ otherwise.  Whichever is charged, we shall prove it is charged at most twice over all Case-5 links.  Since $v$ and $w$ are adjacent children of $u$ when $u$ is deleted, $uv$ and $uw$ are consecutive among the d-links won by $u$.  Link $uw$ is a pairing link, and must be done before or after the assembly pass that does $uv$.  Hence $uv$ is either the first or the last d-link that $u$ wins during this assembly pass. It follows that $xu$ or the deletion is charged at most twice: once for the first d-link $u$ wins during the pass, and once for the last.  Hence there are at most two Case-5 links per real right assembly link, assembly k-link, and deletion.

Combining the bounds for all cases, the number of pairing links is at most one per insertion link (Case 2) and decrease-key link (Case 2); two per pairing f-link (Cases 1 and 4), pairing k-link (Cases 1 and 4), assembly k-link (Case 5), real pairing link (Cases 1 and 3), and real right assembly link (Case 5); and three per deletion (Cases 4 and 5).  By Lemma~\ref{L:f-links} the number of f-links plus the number of deletions is at most the number of insertions, so combining the bounds for insertion links, f-links, and deletions gives a total of at most four times the number of insertions.  Combining the bounds for decrease-key links, pairing k-links, and assembly k-links gives a total of at most three times the number of decrease-keys.  The bound in the theorem follows.     
\end{proof}

\begin{remark}
Even though by Lemma~\ref{L:assembly-links} it is enough to count pairing d-links, it does \emph{not} follow from the proof of Theorem~\ref{T:pairing} that it suffices to count real pairing links. It can happen for example that \emph{all} the pairing links in a deletion are won by temporary nodes and lost by permanent nodes.
\end{remark}

\section{Real right assembly links}
\label{S:size}

To bound the number of real right assembly links, we introduce the notion of node size.  The \emph{size} of $x.s$ of a temporary node $x$ is the number of its descendants, including itself, connected to $x$ by a path of real links.  The size of $x$ is at least one, at most the number of temporary nodes in its subtree, and never decreases, since a real link is cut only when its winner is deleted.  The size of $x$ increases only when $x$ wins a real link with a node $y$, and then it increases by $y.s$.

We shall show that certain real right assembly links cause node size doublings, and use this to bound the number of real right assembly links.  Our argument uses credits and debits.  One credit will pay for one real right assembly link.  We allocate a certain number of credits to each operation. We are allowed when needed to borrow a credit to pay for a link.  This incurs a debit.  We can use credits allocated to later operations to pay off debits.  If there are no debits at the end of a sequence of operations, the sum of the credits allocated to the operations is an upper bound on the number of links.  

Now let's get into the details.  Each time a real right assembly link is done during a sequence of heap operations, we borrow a credit to pay for it.  This incurs a debit, which we give to a temporary node.  To pay off debits, we allocate $\lg n$ credits to each deletion and $(\lg n)/2$ credits to each decrease-key, where $n$ is the number of temporary nodes in the heap at the time of the operation.\footnote{We denote by $\lg$ the base-two logarithm.}

We shall maintain the following \emph{debit invariant}:

\begin{enumerate}
\item[] \textbf{A:} Each temporary node $x$ has at most $(\lg x.s)/2$ debits, except the rightmost root during the assembly pass of a deletion, which if temporary has at most $\lg x.s$ debits.
\end{enumerate}

This invariant gives special status to the rightmost root during the assembly pass of a deletion: Its debit allowance is double the normal amount.

\begin{lemma}
\label{L:assembly-debits}
Any sequence of heap operations maintains A. 
\end{lemma}

\begin{proof}
Invariant A holds initially since there are no debits.  It can fail only when a node with debits is deleted (we must pay off its debits), when an assembly link occurs (incurring a debit) or when an assembly pass ends. (If the remaining root is temporary we must pay off its excess debits.)    

Consider a deletion.  If the root deleted is temporary, we use $(\lg n)/2$ of the credits allocated to the deletion to pay off any debits it has accrued.  This restores A.  The pairing pass of the deletion preserves A, so A holds at the beginning of the assembly pass.  If remaining root $x$ is temporary when the assembly pass ends, we use the remaining $(\lg n)/2$ of the credits allocated to the deletion to pay off any excess debits it was allowed during the assembly pass.  This makes A true after the assembly pass.  It remains to show that assembly links preserve A.

A left assembly link $xy$ preserves A, since the size of $y$ does not change and the size of $x$ does not decrease.

Let $xy$ be a phantom right assembly link.  Whether $x$ is permanent or temporary, $xy$ must be a k-link: If $x$ is permanent, $xy$ cannot be an f-link, since $y$ is eventually deleted; if $x$ and $y$ are both temporary, the phantom link $xy$ is a k-link by definition.  Link $xy$ makes $x$ the rightmost root, so we must pay off the extra $(\lg y.s)/2$ debits allowed to $y$ by A.  For this we use the $(\lg n)/2$ credits allocated to the decrease-key that cuts $xy$, where $n$ is the number of temporary nodes in the tree containing $y$ when the cut happens.  Since the size of $y$ cannot change until $xy$ is cut, $y.s \leq n$.  Thus the decrease-key gives us enough credits to pay off the extra debits, restoring A.

Let $xy$ be a real right assembly link.  This is the crucial case.  We must show that after the link A allows the extra debit incurred by $xy$.  Let variables take their values just before the link.  By A, $x$ and $y$ have at most $(\lg x.s)/2$ and $\lg y.s$ debits before the link, respectively.  The link increases the size of $x$ to $x.s + y.s$ and makes $x$ the rightmost root, so after the link A allows $x$ and $y$ to have a total of $\lg(x.s+ y.s) + \lg(y.s)/2$ debits.  The latter minus the former is
\[\lg(x.s + y.s) - (\lg x.s)/2 - (\lg y.s)/2 \geq 1\]
by the inequality $2\lg(a + b) \geq \lg a + \lg b + 2$.
Thus A holds after the link: We can move debits between $x$ and $y$ as needed to satisfy it.

The lemma holds by induction on the number of deletions, since there are no debits after all temporary nodes have been deleted.
\end{proof}

Lemma~\ref{L:assembly-debits} implies:

\begin{theorem}
\label{T:real-assembly}
The number of real right assembly links is at most $(\lg n)/2$ per decrease-key plus at most $\lg n$ per deletion. 
\end{theorem}

\section{Real pairing links}
\label{S:mass}

Our analysis of real pairing links is similar to our analysis of real assembly links, but instead of showing that certain assembly links cause size doublings, we show that certain pairing links cause halvings of a related parameter, mass.

The \emph{mass} of a real child $x$, denoted $x.m$, is the size of its parent just after $x$ becomes its child.  Equivalently, the mass of a real child $x$ is one plus the sum of its size and those of its real right siblings.  A phantom child has no mass.  Except in the middle of a deletion, a root has no mass. During the assembly and pairing passes of a deletion, the mass of a temporary root is the sum of its size and those of the temporary roots to its right on the root list.

Unlike node size, which never decreases, the mass of a node can increase or decrease.  Nevertheless, by bounding the increases and showing that certain real pairing links cause halvings in mass, we are able to bound the number of real pairing links.  We do not need to use debits, only credits.  One credit will pay for one real pairing link.  Nodes can hold unused credits.  We allocate $(3/2)\lg n + (\lg e)/2$ credits to each deletion and $\lg n$ credits to each decrease-key.

We shall maintain the following invariant:

\begin{enumerate}
\item[] \textbf{B:} Each real child $x$ and each temporary root during a deletion has at least $(\lg x.m)/2$ credits, except during the pairing pass of a deletion, when the leftmost temporary root $x$ not yet involved in a pairing link in the pass has at least $(3\lg x.m)/2$ credits, and during the assembly pass of a deletion, when the rightmost temporary root is not required to have credits.
\end{enumerate}

\begin{lemma}
\label{L:pairing-credits}
Any sequence of heap operations maintains B. 
\end{lemma}
\begin{proof}
Invariant B holds initially.  Only real links, cuts of  real links, phantom pairing links, beginning a pairing pass, or ending an assembly pass can violate B.  We shall show that deletions preserve B and then that decrease-keys and insertion links preserve B.

When a deletion occurs, all children of the root are real.  Deletion of the root decreases by one the mass of each of its children.  These children become the roots on the root list.  Thus deletion of the root preserves B except for the leftmost temporary root if there is one, which now may require extra credits.  To make B true for this root, we give it $\lg n$ of the credits allocated to the deletion.

Suppose B holds just before a pairing link of roots $x$ and $y$, with $x$ left of $y$ before the link.  Let $w$ be the leftmost temporary root not yet involved in a pairing link, if any, and let $z$ be the nearest temporary root right of $y$, if any.  Root $w$ is $x$, $y$, or $z$.  Let variables take their values just before the link.

If the link between $x$ and $y$ is phantom, the link preserves B, since $w.m \geq z.m$ and the link does not increase the mass of any node.

If the link between $x$ and $y$ is real, we need a credit to pay for it.  Let $M$ be the mass of $z$ if $z$ exists, or $1$ if it does not.  Root $w$ is $x$, $x.m = x.s+y.s+M$, and $y.m=y.s+M$.  Before the link, $x$ and $y$ have a total of $(3/2)\lg x.m + (\lg y.m)/2 \geq (3/2)\lg (x.s + y.s + M) + (\lg M)/2$ credits.  After the link, the winner of the link between $x$ and $y$ needs $(\lg x.m)/2 = \lg(x.s+y.s+M)/2$ credits, the loser needs $\lg(x.s + y.s)/2$, and $z$ needs an additional $\lg M$ credits, totaling $\lg(x.s+y.s+M)/2+\lg(x.s + y.s)/2+\lg M$.  Subtracting the number needed after the link from the number available before the link leaves at least
\[\lg(x.s+y.s+M) - \lg(x.s + y.s)/2 - (\lg M)/2 \geq 1\]
by the inequality $2\lg(a + b) \geq \lg a + \lg b + 2$, so the link frees a credit to pay for itself, after we move credits among $x$, $y$, and $z$ as needed to preserve B.  It follows by induction that the pairing pass preserves B.

We do not have to pay for assembly links, but we do have to account for increases in mass caused by such links.  The analysis of the assembly pass is like the proof of Lemma~\ref{L:assembly-debits}.  Invariant B holds at the beginning of the assembly pass. Suppose it holds at some time during the assembly pass.  If the assembly pass is finished, we re-establish B at the remaining root by giving it $(\lg n)/2)$ of the credits allocated to the deletion.  A phantom assembly link preserves B except for the rightmost temporary root: No masses increase, the loser of the link becomes a phantom child, and the rightmost temporary root changes only if the old rightmost root loses the link.

Let $xy$ be a real left assembly link.  The link does not change the mass of $y$ and increases the mass of $x$.  This makes B true for $y$.  Node $x$ remains the rightmost temporary root and so requires no credits.

Let $xy$ be a real right assembly link.  The link increases the mass of $x$, but $x$ becomes the rightmost temporary root, freeing its $(\lg x.m)/2$ credits, where $x.m$ is the mass of $x$ before the link.  The link makes $y$ a real child with mass $x.m$.  We use the credits freed from $x$ to establish B at $y$.

It remains to consider decrease-keys and insertion links.  Cutting a k-link at the beginning of a decrease-key changes no masses and hence preserves B: The loser of the link that is cut was a phantom child and becomes a root, with no mass in either case.  If the decrease-key does a phantom link, the link changes no masses and hence preserves B.  Suppose the decrease-key does a real link $xy$.  Node $y$ becomes a real child, with positive mass.  To establish B for $y$, we give to $y$ $(\lg n)/2$ of the credits allocated to the decrease-key.

A phantom insertion link preserves B.  A real insertion link $xy$ causes the mass of $y$ to increase to at most $n$, the number of temporary nodes in the heap at the time of the link.  Since $y$ is temporary, it eventually becomes the only root in its heap, which must happen before $y$ can be deleted.  Let $n'$ be the number of temporary nodes in the heap the first time after link $xy$ that $y$ becomes the only root in its heap.  To $y$ we give $(\lg n')/2$ of the credits allocated to the operation that caused $y$ to become the only root in its heap.  This operation is either a decrease-key or a deletion, and it is only charged once, since $y$ is the only root in the resulting heap and $xy$ is the most recent previous real insertion link lost by $y$.

If $n' < n$, we need an additional $(\lg n - \lg n')/2$ credits to establish B for $y$  We obtain these from the remaining $(\lg e)/2$ credits allocated to each deletion.  For each deletion, we distribute $(\lg e)/2$ credits equally among the $n'' - 1$ nodes in the heap just after the deletion, $(\lg e)/(2(n'' -1))$ to each.  These credits are spendable at any time, including before the deletion.
 
Between the time $xy$ is created and $y$ next becomes the only root in its heap, there are at least $n - n'$ deletions of other nodes, which give $y$ a total of $(\lg e)/2\sum_{i=n'}^{n-1} 1/i \geq (\lg n - \lg n')/2$ credits, enough to establish B at $y$.

Having covered all the cases, we infer by induction that the lemma is true.
\end{proof}

Lemma~\ref{L:pairing-credits} implies:

\begin{theorem}
\label{T:real-pairing}
The number of real pairing links is at most $\lg n$ per decrease-key plus at most $(3/2)\lg n + (\lg e)/2$ per deletion. 
\end{theorem}

\section{Grand total}
\label{S:total}

Theorems~\ref{T:pairing}, \ref{T:real-assembly}, and \ref{T:real-pairing}, give a bound on the number of pairing links of $4$ per insertion, $3\lg n + 3$ per decrease-key, and $5\lg n + \lg e$ per deletion.  By Lemma~\ref{L:assembly-links}, the same bound applies to assembly links.  Adding the number of insertion links (one per insertion, by Lemma~\ref{L:insertion-links}) and decrease-key links (one per decrease-key), we obtain:

\begin{theorem}
\label{T:total}
The total number of links done by a sequence of pairing heap operations is at most $9$ per insertion, $6\lg n + 7$ per decrease-key, and $10\lg n + 2\lg e$ per deletion.
\end{theorem}

\begin{corollary}
The running time of a sequence of pairing heap operations is $\OO(1)$ per insertion, zero per meld, and $\OO(\log n)$ per decrease-key and deletion.
\end{corollary}

\section{Remarks}
\label{S:remarks}
One outcome of our analysis is that each decrease-key that does a k-link or an f-link takes $\OO(1)$ amortized time, not $\OO(\log n)$.  In particular, decreasing the key of a permanent node takes $\OO(1)$ amortized time.  Furthermore, the bounds for decrease-key and deletion are in terms of the number of temporary nodes in the heap, not of all nodes.  

If one is willing to allow insertions and melds to take $\OO(\log n)$ amortized time, our analysis can be significantly simplified: We do not need the concept of permanent and temporary nodes, we do not need Theorems~\ref{T:pairing} and \ref{T:real-assembly}, we can use Lemma~\ref{L:assembly-links} directly, and we can improve the constants in Theorem~\ref{T:real-pairing}: In the proof of that theorem, we charge the increase in the mass of the loser of an insertion link to the insertion or meld that does the link.  This gives us the following alternative to Theorem~\ref{T:total}:

\begin{theorem}
\label{T:total-alt}
The total number of links done by a sequence of pairing heap operations is at most $\lg n + 1$ per insertion, $\lg n$ per meld, $\lg n + 3$ per decrease-key, and $2\lg n$ per deletion.
\end{theorem}

One can combine the analyses of Sections~\ref{S:size} and~\ref{S:mass} by using a single credit invariant in place of A and B (with appropriate exceptions in the middle of pairing and assembly passes):

\begin{enumerate}
\item[] \textbf{C:} Every node $x$ that is a real child or temporary root in the middle of a deletion has at least $(\lg x.m - \lg x.s)/2 = \lg(x.m/x.s)/2$ credits.
\end{enumerate}

We have not done this, for two reasons: the analyses of real right assembly links and real pairing links are in fact independent, and (we hope) easier to understand separately rather than together; and the former requires only sizes, not masses.  But the idea of using $\lg(x.m/x.s)$ or some related function of $x.m/x.s$ as a potential gives much more.  By using a function of $x.m/x.s$ that grows more slowly than $\log$, Pettie~\cite{PettiePairing} obtained an amortized time bound for pairing heaps of $\OO(2^{2\sqrt{\lg\lg \nm}})$ per insertion, meld, and decrease-key, where $\nm$ is an upper bound on the size of any heap over all operations, while preserving the $\OO(\log n)$ bound per deletion.  The first author~\cite{phd2022} improved Pettie's bound for insert, meld, and decrease-key to $\OO(\sqrt{\log\log \nm} 2^{\sqrt{2\lg\lg \nm}})$.  We have not yet verified whether our ideas combine with this result to reduce the bound per insertion and meld to $\OO(1)$.

The ``natural" function to use as a potential is $\lg\lg(x.m/x.s)$.  Using this function in a novel analysis along with the concept of temporary and permanent nodes, we have obtained significantly improved bounds for three other kinds of self-adjusting heaps: \emph{multipass} pairing heaps~\cite{FSST86}, in which deletions do repeated paring passes until there is only one root; and slim and smooth heaps~\cite{KS19, HKST21}, which do deletions using locally maximum linking.  Our amortized time bounds for multipass pairing heaps are $\OO(\log n)$ for deletion, $\OO((\log\log n)(\log\log\log n))$ for decrease-key, and $\OO(1)$ for the other operations~\cite{phd2022}. Our bounds for slim and smooth heaps are the same but without the extra $\log\log\log n$ factor in the decrease-key bound~\cite{phd2022}: slim and smooth heaps have efficiency that matches the lower bounds.

Of course, the big open question is whether pairing heaps match the lower bounds.  Being optimists, we conjecture that the answer is yes.

\bibliographystyle{plain}
\bibliography{SH}
\end{document}